\documentclass[journal]{IEEEtran}

\ifCLASSINFOpdf
\else
   \usepackage[dvips]{graphicx}
\fi
\usepackage{url}

\usepackage{graphicx}
\usepackage{amsmath,amsthm,amssymb}
\usepackage{xcolor}
\usepackage{nicefrac}
\usepackage{cancel}
\usepackage{hyperref}
\usepackage{microtype}
\usepackage[full]{textcomp}
\usepackage[scaled=.98,sups,osf]{XCharter}%
\usepackage[scaled=1.04,varqu,varl]{inconsolata}%
\usepackage[type1,scaled=0.98]{cabin}%
\usepackage[utopia,vvarbb,scaled=1.07]{newtxmath}
\usepackage[cal=cm]{mathalfa}
\usepackage{amsmath,amssymb,dsfont,amsthm}
\usepackage{stmaryrd}

\newtheorem{theorem}{Theorem}

\newtheorem{lemma}{Lemma}

\theoremstyle{definition}

\newtheorem{assumption}{Assumption}

\theoremstyle{remark}

\DeclareSymbolFont{sansops}{OT1}{\sfdefault}{m}{n}
\SetSymbolFont{sansops}{bold}{OT1}{\sfdefault}{b}{n}

\makeatletter
\renewcommand\operator@font{\mathgroup\symsansops}
\makeatother

\DeclareSymbolFont{sfoperators}{OT1}{cmss}{m}{n}
\DeclareSymbolFontAlphabet{\mathsf}{sfoperators}

\makeatletter
\def\operator@font{\mathgroup\symsfoperators}
\makeatother

\newcommand{\R}{\mathbb{R}}

\newcommand{\E}{\operatorname{ E}}

\newcommand{\vct}[1]{\boldsymbol{#1}}
\newcommand{\mtx}[1]{\boldsymbol{#1}}
\newcommand{\diag}{\operatorname{diag}}

\newcommand{\<}{\langle}
\renewcommand{\>}{\rangle}

\newcommand{\T}{\top}

\newcommand{\ip}[2]{\left\<#1, #2\right\>}
\newcommand{\norm}[1]{\left|\left|#1\right|\right|}

\newcommand{\Expec}[2][]{\E_{#1}\left[#2\right]}

\newcommand{\p}[1]{\left(#1\right)}

\renewcommand{\c}[1]{\left\{#1\right\}}

\newcommand{\set}[1]{\mathcal{#1}}

\newcommand{\linop}[1]{\mathcal{#1}}    

\DeclareMathOperator*{\argmin}{\operatorname{arg~min}}

\newcommand{\va}{\vct{a}}

\newcommand{\ve}{\vct{e}}

\newcommand{\vg}{\vct{g}}

\newcommand{\vp}{\vct{p}}
\newcommand{\vq}{\vct{q}}
\newcommand{\vr}{\vct{r}}

\newcommand{\vu}{\vct{u}}
\renewcommand{\vv}{\vct{v}}
\newcommand{\vw}{\vct{w}}
\newcommand{\vx}{\vct{x}}

\newcommand{\vz}{\vct{z}}

\newcommand{\vone}{\vct{1}}

\newcommand{\mA}{\mtx{A}}

\newcommand{\mC}{\mtx{C}}

\newcommand{\mE}{\mtx{E}}

\newcommand{\mP}{\mtx{P}}

\newcommand{\mR}{\mtx{R}}

\newcommand{\mV}{\mtx{V}}

\newcommand{\mX}{\mtx{X}}
\newcommand{\mY}{\mtx{Y}}
\newcommand{\mZ}{\mtx{Z}}

\newcommand{\mId}{\mathbf{I}}

\newcommand{\loQ}{\linop{Q}}

\newcommand{\setD}{\set{D}}

\newcommand{\setK}{\set{K}}

\newcommand{\setN}{\set{N}}
\newcommand{\setO}{\set{O}}

\newcommand{\setQ}{\set{Q}}
\newcommand{\setR}{\set{R}}

\newcommand{\setT}{\set{T}}

\newcommand{\sgn}{\operatorname{sign}}

\newcommand{\dist}[1]{\operatorname{\sf #1}}
\newcommand{\dUnif}{\dist{Uniform}}

\newcommand{\floor}[1]{\left\lfloor #1 \right\rfloor}

\title{Error Bounds for Radial Network Topology Learning from Quantized Measurements}
\author{Samuel Talkington, Aditya Rangarajan, Pedro A. de Alc\^antara, \\Line Roald, Daniel K. Molzahn, Daniel R. Fuhrmann 
\thanks{The authors gratefully acknowledge funding from PSERC project T-67: Smart Meter-Driven Distribution Grid Visibility and Control. The work of S. Talkington is supported by the National Science Foundation Graduate Research Fellowship Program under Grant No. DGE-1650044. 
Any opinions, findings, and conclusions or recommendations expressed in this material are those of the author(s) and do not necessarily reflect the views of the National Science Foundation.
}
}

\IEEEaftertitletext{\vspace{-2\baselineskip}}
\begin{document}

\maketitle
\begin{abstract}
   We probabilistically bound the error of a solution to a radial network topology learning problem where both connectivity and line parameters are estimated. In our model, data errors are introduced by the precision of the sensors, i.e., quantization. This produces a nonlinear measurement model that embeds the \textit{operation of the sensor communication network} into the learning problem, expanding beyond the additive noise models typically seen in power system estimation algorithms. We show that the error of a learned radial network topology is proportional to the quantization bin width and grows sublinearly in the number of nodes, provided that the number of samples per node is logarithmic in the number of nodes.
\end{abstract}

\vspace{-1em}
\section{Introduction}
\label{sec:intro}



Efficiently allocating \textit{limited smart meter bandwidth} is an emerging challenge at the intersection of power and communication engineering \cite{smart_meter_pinging,rangarajan_gupta_molzahn_roald-forecast_dsse}. As smart meters proliferate, power engineers will increasingly rely on analog-to-digital \textit{quantization} methods to improve computational and communication efficiency by mapping continuous measurements to discrete intervals \cite{gray_quantization_survey_1998}; see Fig.~\ref{fig:topology-learning-highlevel} (left). While coarse quantization accelerates computations \cite{gholami2022survey}, it introduces \textit{nonlinear}, typically \textit{non-Gaussian}, measurement noise \cite{plan_generalized_2016,thrampoulidis_quantized_lasso_2020}, contrasting with the additive Gaussian assumptions common in power system estimation.

To help address these challenges, we present a new \textit{sample complexity} analysis method for distribution topology learning, a well-studied inference task in power engineering; see \cite{deka_learning_2024} for a comprehensive review. Sample complexity refers to the \textit{number of measurements needed} to ensure that the \textit{error of an estimated parameter is bounded} by a chosen tolerance \cite{vapnik2013nature}. In particular, this work provides such topology learning error bounds under the effects of practically relevant communication non-idealities. 

To the knowledge of the authors, these contributions are the first of their kind and expand upon recent deterministic studies of measurement requirements for topology learning \cite{rin2025electricgridtopologyadmittance}. Concretely, we address the following question:
\begin{quote}
    \textit{Consider an $n$-node distribution network with unknown topology (connectivity and parameters). Suppose that we collect measurements at every node with a uniform quantization bin width $\Delta>0$.  How many samples per node $s = m/n$ are needed to recover the topology of the network from the $m=sn$ quantized measurements, up to a desired error tolerance?}
\end{quote}
This research question is illustrated in Fig.~\ref{fig:topology-learning-highlevel}  (right), outlined in Section~\ref{sec:form}. We give a precise answer in Section~\ref{sec:results} by providing a bound on the error of a topology estimate given a prescribed number of measurements.

\begin{figure}[t]
    \centering
    \includegraphics[width=0.5\linewidth]{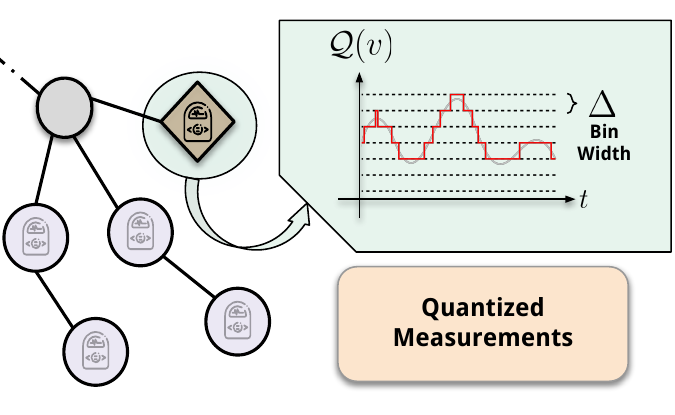}\hfill
    \includegraphics[width=0.47\linewidth]{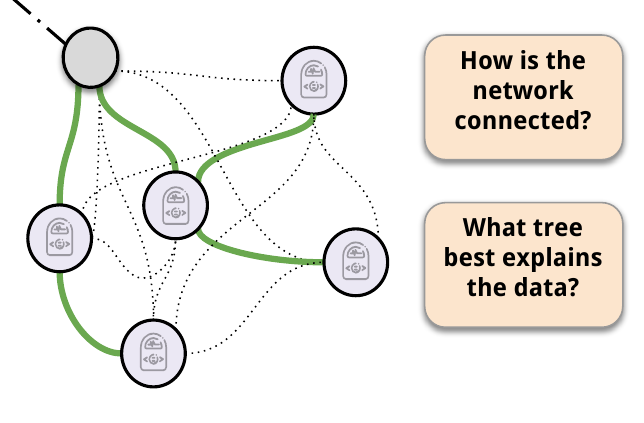}
    \caption{Conceptual illustration of the proposed problem. Left: Distribution network smart meter measurements have realistic quantization noise corresponding to \emph{bin width} $\Delta>0$. Right: Estimate the topology \textit{and} parameters of the distribution network (green, solid lines).}
    \label{fig:topology-learning-highlevel}
\end{figure}


\section{Problem formulation}
\label{sec:form}


\vspace{-0.05em}
\subsection{Communication model}
\label{sec:measure}

We wish to recover an unknown vector of line parameters $\vw_\star \in \R^d$ from $m \geq d$ \textit{quantized measurements} of the form
\begin{equation}
    \label{eq:quant_meas}
    p_i \triangleq \loQ(\ip{\va_i}{\vw_\star} ), \quad i=1,\ldots,m, 
\end{equation}
where $\loQ : \R \to \R$ is a nonlinear \emph{quantization function} and $\p{p_i,\va_i}$, $i=1,\ldots,m$ are measurements collected from smart meters (both spatially \textit{and} temporally). Vectors $\c{\va_i}$ are rows of a particular sensing matrix $\mA \in \R^{m \times d}$ and contain pairwise voltage differences measured co-temporally with active power injection $p_i$. We will define matrix $\mA$ explicitly in Section \ref{sec:form_learning} as a consequence of linearizing the power flow equations. In the same way, we will argue that the sparsity pattern of $\vw_\star$ represents the network topology (cf.  \cite[Sec. IV-B]{deka_learning_2024}).

Concretely, we focus on the setting in which coarsely quantized measurements \eqref{eq:quant_meas} are generated from a uniformly dithered quantization function (cf. \cite{gray_quantization_survey_1998,thrampoulidis_quantized_lasso_2020}) with bin width $\Delta >0$. These measurements take the form
\begin{equation}
    \label{eq:unif_quant_measmt}
    p_i = \Delta\cdot\p{\floor{\frac{\ip{\va_i}{\vw_\star} + \tau_i}{\Delta}} + \frac{1}{2}} \quad i=1,\ldots,m,
\end{equation}
where $\tau_i \sim \dUnif\p{-\frac{\Delta}{2},\frac{\Delta}{2}}$, and $\Expec[]{p_i}=\ip{\va_i}{\vw_\star}$. The \textit{dither} $\tau_i$ is a purposely applied random noise component that is generated by a sensor and added to an input signal prior to
quantization. This technique is rather well-established and commonly used by sensors in practice, and by statisticians in theory; see \cite{gray_dithered_quantizer_1993,gray_quantization_survey_1998}.

\vspace{-1em}
\subsection{Statistical tools}
\label{sec:form:math}
Suppose that $\vw_\star$ lies in a convex constraint set $\setK \subseteq \R^d$ that encodes some sort of \emph{known structure} of the parameter $\vw_\star$, e.g., knowledge of where certain lines are located (but not their switching status), or whether the network is operated radially. We can solve for a constrained estimate $\hat\vw \in \R^d$ that obeys this structure via the program \eqref{eq:g-lasso}, known as the \textit{generalized LASSO}:
\begin{equation}
\label{eq:g-lasso}
    \hat{\vw} = \argmin_{\vw \in \setK} \frac{1}{2m} \sum_{i=1}^m \p{\ip{\va_i}{\vw} - p_i}^2.
\end{equation}
The quality of the estimate produced by the program \eqref{eq:g-lasso} can be quantified with a statistical tool known as the \textit{Gaussian width}. We introduce this tool briefly and refer the reader to \cite{chandrasekaran_convex_2012}, \cite[Ch. 7.5]{Vershynin_2018} and \cite{plan_generalized_2016,thrampoulidis_quantized_lasso_2020} for further technical information. Formally, the Gaussian width $\omega(\setT)$ of a set $\setT \subseteq \R^d$ is defined as
\begin{equation}
\omega(\setT) \triangleq \Expec[]{\sup_{\vu\in \setT} \ip{\vu}{\vg}},
\end{equation}
where $\vg \sim \setN(0,\mId)$ is a vector of i.i.d. standard Gaussians. The Gaussian width $\omega$ is a useful tool for predicting the behavior of structured convex problems with random input data \cite{chandrasekaran_convex_2012}. It represents the ``size'' of a standard Gaussian process over $\setT$, which we can use as a stochastic \textit{comparison} for our structured problem. Two classic examples of structured constraint sets with well-behaved Gaussian widths are sparse and low-rank constraint sets \cite{chandrasekaran_convex_2012}.

The squared Gaussian width, $\omega^2(\setT)$, provides a measure of the effective dimension of the set $\setT$; informally, this represents the minimum number of orthogonal directions required to capture the variation or structure within $\setT$. We will also need another object for our analysis\textemdash the \textit{tangent cone} of a set $\setK \subset \R^d$ at $\vx \in \R^d$, which is defined as
\begin{equation}
\label{eq:tangent-cone}
\setD(\setK,\vx) \triangleq \c{\lambda \vu \: : \: \lambda \geq 0, \, \vu \in \setK - \vx},
\end{equation}
where $\setK-\vx$ is the set $\setK$ translated by $\vx$. The tangent cone \eqref{eq:tangent-cone} generalizes the idea of the tangent space to a nonlinear surface that may have non-differentiable points.

\subsection{Network and power flow model}
\label{sec:form_network_and_pf}
The abstract quantities in \eqref{eq:quant_meas} and \eqref{eq:unif_quant_measmt} can be applied to the Linear Coupled Power Flow (LCPF) model \cite{deka_learning_2024,bolognani_fast_2015}. As the focus of this work is on sample complexity, we make the following simplifying assumption on the reactive power.
\begin{assumption}[Fixed power factor]
    \label{assum:fixed_pf}
    Let $\vp \in \R^n$ be active power injections.
    Assume that reactive power injections satisfy $\vq = \kappa \vp$, where $\kappa \in \R$ is a known constant.
\end{assumption}
Assumption \ref{assum:fixed_pf} is equivalent to defining a fixed power factor $\phi \in (0,1]$ such that $\kappa = \pm \phi^{-1}\sqrt{1-\phi^2}$, where $\sgn(\kappa) = 1$ if the injections are inductive and $\sgn(\kappa)=-1$ if the injections are capacitive. For more information on the consequences of this assumption, see \cite{stanojev_tractable_2023,arghandeh_chapter_2024}.

\subsubsection{Linear Coupled Power Flow (LCPF) model}

For a radial network with $n+1$ nodes, let $\mC \in \c{-1,0,1}^{n \times n}$ be the invertible, reduced, branch-to-node incidence matrix of a tree graph with the column corresponding to the slack node removed. Let $\c{\p{\vv_t,\vp_t}}_{t=1}^m$ be a sequence of nodal voltage magnitude and active power measurements. Under Assumption \ref{assum:fixed_pf}, the voltage magnitudes satisfy~\cite{deka_learning_2024,bolognani_fast_2015}:
\begin{equation}
\vv_t - \vone = \p{\mR + \mX \cdot \kappa} \vp_t, \quad t=1,\ldots,m
\end{equation}
where $\mR + \mX  \kappa \triangleq \mZ$ is the (reduced) \emph{equivalent impedance matrix}. We can write $\mZ$ as the inverse of a \textit{real-valued} equivalent admittance matrix $ \mY \in \R^{n \times n}$, where
\begin{equation}
\mZ = \mC^{-1} \diag\p{\vr + \kappa \vx} \mC^{-\T}\triangleq \mY^{-1},
%
%
\end{equation} 
where $\vz \triangleq \vr + \kappa \vx \in \R^{n}$ are the line impedances (scaled). The matrix $\mC^{-1}  \in \c{-1,0}^{n \times n}$ is a lower triangular matrix where the non-zero entries of each column $j$ represent the descendants of node $j$ in the tree. In particular, over the indices $i,j=1,\ldots,n$, we have $\p{\mC^{-1}}_{ij} =-1$ if $i=j$ or $i$ is a descendant of $j$, and $\p{\mC^{-1}}_{ij}=0$ otherwise.




\subsubsection{Lifting to the complete graph}
\newcommand{\mCtilde}{\mtx{\tilde{C}}}
\newcommand{\mYtilde}{\mtx{\tilde{Y}}}
 
 Echoing \cite{rin2025electricgridtopologyadmittance}, we can view the problem of recovering $\vw_\star$ as \textit{sparsifying} a complete undirected graph $K_{n+1}$ with $n+1$ nodes and ${n+1 \choose 2}$ lines. Accordingly, let $\mYtilde(\cdot)$ be the \textit{admittance matrix operator} that sends line parameters $\vw \in \R^{n +1 \choose 2}$ to a corresponding (non-reduced) $(n+1) \times (n+1)$ admittance matrix: 
    \begin{equation}
    \mYtilde(\vw) = \mCtilde^\T \diag(\vw) \mCtilde 
    = \sum_{i=1}^{n} \sum_{j=i+1}^{n} \mE_{ij}w_{ij}.
    \end{equation}
    Here,  $\mE_{ij} \triangleq \p{\ve_i - \ve_j}\p{\ve_i - \ve_j}^\T$ is an \textit{elementary Laplacian matrix}, where $\ve_i$ is the $i$-th standard basis vector in $\R^{n+1}$. 
    The matrix $\mCtilde \in \c{-1,0,1}^{{n +1\choose 2} \times n+1}$ is the incidence matrix that corresponds to a complete graph. We can now represent the set of all radial networks as the set of all line parameter vectors $\vw$ that correspond to a connected radial network with $n+1$ nodes. This set takes the form
    \begin{equation}
    \label{eq:spanning_tree_frame}
    \setR = \c{ \vw \in \R^{n+1 \choose 2} \: : \: \norm{\vw}_0 = n, \ \ \lambda_2(\mYtilde(\vw))>0}.  
    \end{equation}
The sparsity condition $\norm{\vw}_0=n$ ensures $n$ lines, and the condition on the second smallest eigenvalue $\lambda_2(\mYtilde(\vw))>0$ means the graph is connected, i.e., $\vx^\T \mYtilde(\vw)\vx > 0$ for all $\vx \perp \vone$. The constraint set \eqref{eq:spanning_tree_frame} is intractable; however, in Section \ref{sec:form_learning} we provide a relaxation that performs well both experimentally and theoretically. Hereafter, we implicitly remove the slack node column of $\mCtilde$ in our calculations, and refer to the network as having $n$ nodes.



%
%
%
%

\section{Main result}
\label{sec:results}

    \subsection{Topology learning problem}
    \label{sec:form_learning}

    Let $\mP,\mV \in \R^{n \times s}$ be data matrices whose columns are $s$ samples of active power and voltage magnitude measurement vectors, respectively, across all $n$ nodes. We seek to recover an $n$-sparse vector
         $\vw_\star$ of line parameters, where $w_{\star,i}=\nicefrac{1}{z_{\star,i}}$ if $z_{\star,i} \neq0$, and $0$ otherwise. This problem can be written as a sparse recovery problem over the \emph{complete graph} by defining the measurement system
          \begin{equation}
    \label{eq:topology_est_complete_graph}
         \vp = \setQ\p{\operatorname{\sf vec}\p{\mP}} = \setQ\p{\mA \vw_\star}, \quad \vp \in \R^{sn}.
    \end{equation}
    The $sn \times {n+1 \choose 2}$ sensing matrix is 
    $
    \mA \triangleq  \mV^\T\mCtilde^\T \odot \mCtilde^\T,
    $
     where $\odot$ is the Khatri-Rao matrix product. 
        

\vspace{-0.25cm}
\subsection{Bounding the error of the parameter estimate}
\label{sec:results_theory}

    Taking advantage of the sparsity, we can characterize the Gaussian width of the tree set \eqref{eq:spanning_tree_frame}.
    \begin{lemma}
        \label{lemma:gaussian_width_n_sparse}
     Define the set $\setT_{\setR,\vw_\star} \triangleq \setD(\setR,\vw_\star) \cap \mathbb{S}^{d-1}$ as the intersection of the tangent cone \eqref{eq:tangent-cone} of constraint set $\setR$ with the unit shell $\mathbb{S}^{d-1} \triangleq \c{\vu \in \R^d \: : \: \norm{\vu}_2 =1}$, $d={n+1 \choose 2}$. The squared Gaussian width of the set of radial networks \eqref{eq:spanning_tree_frame} satisfies $\omega^2\p{\setT_{\setR,\vw_\star}} \leq 2n\log\p{\frac{n+1}{2}} + \frac{3}{2}n.$
    \end{lemma}
    \begin{proof}
        The ground truth $\vw_\star$ corresponds to a spanning tree of the complete graph with $n+1$ nodes. Since this tree must have $n$ edges, vector $\vw_\star$ must be $n$-sparse, i.e., $\norm{\vw_\star}_0 = n$. Dropping connectivity, define the relaxation $\setK \supseteq \setR$ as
        \begin{equation}
        \label{eq:relaxed_constraint_set}
        \setK \triangleq \c{\vw \in \R^{{n +1\choose 2}} : \norm{\vw}_1 \leq \norm{\vw_\star}_1},
        \end{equation}
        i.e., the $\ell_1$-norm ball of radius $\norm{\vw_\star}_1$.
        It is well known (cf. \cite{chandrasekaran_convex_2012,thrampoulidis_quantized_lasso_2020}) that, as $\vw_\star$ is $n$-sparse, we can conclude that $\omega^2(\setT_{\setR,\vw_\star}) \leq \omega^2(\setT_{\setK,\vw_\star}) \leq 2n \log(\frac{n+1}{2}) + \frac{3}{2}n$.
    \end{proof}
    Using Lemma \ref{lemma:gaussian_width_n_sparse}, the error of the topology learned by \eqref{eq:g-lasso} with constraint set \eqref{eq:relaxed_constraint_set} and bin width $\Delta >0$ is bounded and can be quantified.
    
\begin{theorem}
    \label{thm:error_bound_unif_quantization}
    Suppose that $s = \frac{m}{n}$ samples with bin width $\Delta >0$ are collected throughout an $n$-node distribution network, and suppose that the $\c{\va_i}$ are any i.i.d. sub-Gaussian random vectors. Then, there exist constants $C,c_1,c_2>0$ such that the error of the topology estimate from solving \eqref{eq:g-lasso}, with probability at least 0.99 (i.e., 99\%) is bounded as
    \begin{equation}
    \label{eq:topology_error_bound}
    \norm{\hat{\vw} - \vw_\star}_2 \leq C\Delta \cdot \sqrt{\frac{2\log\p{\frac{n+1}{2}} + \nicefrac{3}{2} }{s}},
    \end{equation}
    provided the number of samples $s \geq \setO(\log n).$
\end{theorem}
\begin{proof}

    Combining Lemma \ref{lemma:gaussian_width_n_sparse} with \eqref{eq:glasso_error_bound} and applying \cite[Thm. III.1]{thrampoulidis_quantized_lasso_2020}, there exists a constant $C>0$ such that the minimizer $\hat\vw$ of the program \eqref{eq:g-lasso} satisfies, with probability at least 0.99,\begin{subequations}
    \label{eq:glasso_error_bound}
    \begin{align}
        \norm{\hat{\vw} - \vw_\star}_2 &\leq C  \Delta \cdot \frac{\omega\p{\setT_{\setK,\vw_\star}}}{\sqrt{m}},\\
        &\leq C\Delta \cdot \sqrt{\frac{2\cancel{n}\log(\frac{n+1}{2}) +\frac{3}{2}\cancel{n}}{s\cancel{n}}},
    \end{align}
\end{subequations}
provided the total number of measurements $m = sn \geq c_1\omega^2\p{\setT_{\setK,\vw_\star}} +c_2$ for some absolute constants $c_1,c_2>0$. Due to Lemma \ref{lemma:gaussian_width_n_sparse}, if the number of samples per node $s \geq c_1(2\log\p{\frac{n+1}{2}} + \nicefrac{3}{2})+c_2=\setO(\log n)$ we satisfy this measurement requirement.
\end{proof}
%


\begin{figure}[t]
    \centering
    \includegraphics[width=0.95\linewidth]{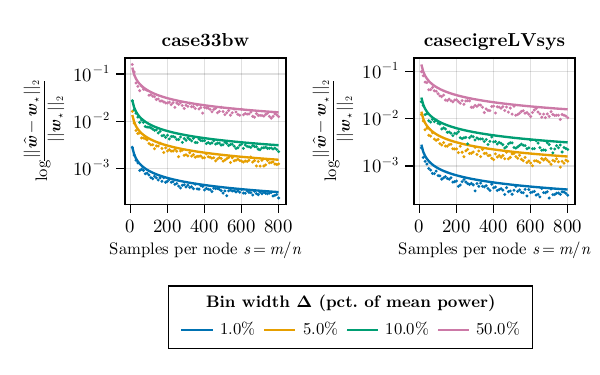}
    \caption{Relative error vs. number of samples for repeated runs of the topology learning program \eqref{eq:g-lasso} on the Baran \& Wu (left) and CIGRE low-voltage (right) test cases. The scatter markers show the experimental error obtained for one run of \eqref{eq:g-lasso}, where each color corresponds to a bin width $\Delta$. The solid lines show the error prescribed by \eqref{eq:topology_error_bound}.}
    \label{fig:quantized_est}
\end{figure}

\vspace{-1.5em}
\subsection{Numerical results}\label{sec:numerics}
Fig.~\ref{fig:quantized_est} compares the error bound~\eqref{eq:topology_error_bound} with the parameter error obtained by numerically solving many instances of the program~\eqref{eq:g-lasso}. The errors are plotted against the number of samples per node, varied over 100 discrete uniformly spaced points $10 \leq s \leq 800$. For each instance of the program~\eqref{eq:g-lasso}, a random matrix $\mV\in\R^{n\times s}$ of independent Gaussian voltages is generated, where the mean of each column is the AC power flow solution of the feeder, and the variance of each column entry is 10\% thereof. Active power measurements are then generated as in~\eqref{eq:topology_est_complete_graph}. The quantization bin width $\Delta$ is varied as a percentage of the sample mean of the absolute active power injections $\frac{1}{m}\norm{\mA \vw_\star}_1$. The solid curves in Fig.~\ref{fig:quantized_est} depict the error bounds given by~\eqref{eq:topology_error_bound}, for each bin width $\Delta$. 

Following \cite{thrampoulidis_quantized_lasso_2020}, the constant $C$ is not specified by the theorem. 
Therefore, we numerically select $C$ for each feeder by computing the maximum ratio of empirical errors and~\eqref{eq:topology_error_bound} over all $s,\Delta$. We found $C \approx 13$ for \texttt{case33bw} and $C \approx 10$ for \texttt{casecigreLV}, which we use for the solid curves in Fig.~\ref{fig:quantized_est}. The scaling $\setO(\Delta \sqrt{\nicefrac{\log n}{s}})$ predicted by~\eqref{eq:topology_error_bound} matches the experimental errors across all $\Delta$.

\vspace{-0.6em}
\section{Conclusion}
\label{sec:conclusion}

    

We provided a framework to \textit{predict the error of a learned distribution network topology\textemdash before solving any optimization problem.} By exploiting the radial structure, we achieved an error \textit{sublinear} in the number of nodes, up to constant factors determined by the precision of the sensor and the underlying probability distributions. Formally, we showed that the topology \textit{and} line parameters of a radial network can be recovered up to a relative error of $\setO(\Delta\sqrt{\nicefrac{\log n}{s}})$, given a collection of $\setO(\log n)$ samples from every node with quantization bin width $\Delta>0$. Thus, before collecting data, one can size quantization and sampling budgets to hit a target topology error. Note that~\eqref{eq:glasso_error_bound} is determined by the bin width $\Delta$. Deriving bounds involving communication bit rates would require assumptions on the variance of the underlying measurement being quantized.


\bibliographystyle{ieeetr}
{\footnotesize
\bibliography{Refs}
}






\end{document}